\documentclass{article} % For LaTeX2e
\usepackage{arxiv,times}
\usepackage{xspace}
\usepackage{amsmath}
\usepackage{amssymb}
\usepackage{mathtools}
\usepackage{amsthm}
\usepackage{bm}

\usepackage{amsmath,amsfonts,bm}

\def\eqref#1{equation~\ref{#1}}
\def\1{\bm{1}}

\def\rvX{{\mathbf{X}}}

\DeclareMathAlphabet{\mathsfit}{\encodingdefault}{\sfdefault}{m}{sl}
\SetMathAlphabet{\mathsfit}{bold}{\encodingdefault}{\sfdefault}{bx}{n}

\usepackage{hyperref}
\usepackage{cleveref}
\usepackage{url}
\newtheorem{theorem}{Theorem}

\newtheorem{lemma}{Lemma}

\newcommand{\sch}[0]{Schrödinger\xspace}

\title{\(\tilde {\mathcal  O}(N^2)\) Representation of General Continuous Anti-symmetric Function}

\author{%
  Haotian Ye$^{1,}$\thanks{Equal contributions.} \quad Ruichen Li $^{2,*}$ \quad Yuntian Gu$^{3,*}$  \quad Yiping Lu$^{1}$\quad Di He$^{2,}$\thanks{Corresponding to: \texttt{\{dihe,wanglw\}@pku.edu.cn}.} \quad Liwei Wang$^{2,4,\dag}$\\
  $^1$Stanford University\\
  $^2$National Key Laboratory of General Artificial Intelligence, Peking University\\
  $^3$Yuanpei College, Peking University\quad\\
  $^4$Center for Machine Learning Research, Peking University\\
}

\iclrfinalcopy % Uncomment for camera-ready version, but NOT for submission.
\begin{document}

\maketitle

\begin{abstract}
% This document provides a basic paper template and submission guidelines.
% Abstracts must be a single paragraph, ideally between 4--6 sentences long.
% Gross violations will trigger corrections at the camera-ready phase.
In quantum mechanics, the wave function of fermion systems such as many-body electron systems are anti-symmetric (AS) and continuous, and it is crucial yet challenging to find an ansatz to represent them. 
This paper addresses this challenge by presenting an \(\tilde{\mathcal{O}}(N^2)\) ansatz based on permutation-equivariant functions. We prove that our ansatz can represent any AS continuous functions, and can accommodate the determinant-based structure proposed by \citet{hutter2020representing}, solving the proposed \textit{open problems} that \(\mathcal{O}(N)\) Slater determinants are sufficient to provide universal representation of  AS continuous functions. 
% From a theoretical side, the key lemma we prove can further help improve the efficiency of \citet{chen2023exact} from \(\mathcal{O}(N^5)\) to $\mathcal O(N^4)$, and our ansatz generalizes beyond \citet{pang2022n} that can only represent ground state wave functions. 
Together, we offer a generalizable and efficient approach to representing AS continuous functions, shedding light on designing neural networks to learn wave functions.
% This paper addresses the challenge of representing continuous anti-symmetric (AS) functions, which are crucial in fermion systems in quantum chemistry. We present an \(\tilde{\mathcal{O}}(N^2)\) ansatz using permutation-equivariant multi-electron functions, proving its universal continuous anti-symmetric approximation capability. This significantly surpasses the previously most efficient \(\mathcal{O}(N^5)\) ansatz. Furthermore, we resolve a longstanding open problem, showing that \(\mathcal{O}(N)\) Slater determinants can achieve universal continuous anti-symmetric representation. Our approach offers a more computationally efficient solution than existing methods, contributing significantly to the field of computational quantum physics and potentially impacting material science and drug discovery.
\end{abstract}

\section{Introduction}
% \lrc{Note: Intuitively, I want to imply that the importance of developing an anti-symmetric function is the same as solving the curse of dimension. This is achieved through linking the sign problem with the curse of dimensionality.}
%Quantum mechanics, central in describing microscopic phenomena, has been understood for nearly a century, yet solving quantum many-body systems remains a challenge. Solving the time-independent electronic Schrödinger equation is key for understanding chemical properties, with significant implications in material science and drug discovery. \citep{heifetz2020quantum, keimer2017physics}. 

Quantum mechanics, central in describing microscopic phenomena such as superconductivity and quantum Hall effect, is essential and has been investigated for nearly a century \citep{feynman1965feynman}. Along its development, one of the important yet highly challenging problems is to compute the \textit{wave function} by solving the time-independent \sch equation of a given system. Since wave functions completely determine the initial state of quantum systems and can be used to derive subsequent states at any time, it is important in understanding microscopic systems and significant in various implications in material science and drug discovery \citep{heifetz2020quantum, keimer2017physics}. 

%Quantum Monte Carlo (QMC) \citep{needs2020variational, foulkes2001quantum} is a family of \emph{ab initio} techniques based on random sampling that provides accurate approximations for calculating the physical properties of real systems. Among these, variation Monte Carlo (VMC) involves minimizing the energy with respect to the trail wave function. The accuracy of the solution heavily relies on the choice of ansatz, which limits its performance compared to other deterministic approaches, such as CCSD(T).

Solving the \sch equation is highly nontrivial, since the dimension of the wave function scales linearly with the number of particles, leading to the typical curse of dimensionality (CoD) \citep{han2018solving}. Additionally, wave functions are hard to express as they do not have analytical expression, even in a simple atomic system with two electrons \citep{aznabaev2018nonrelativistic}. To overcome these challenges, Quantum Monte Carlo (QMC) \citep{needs2020variational, foulkes2001quantum} employs random sampling to approximate integrals in high dimensional spaces, avoiding solving the \sch equation directly and releasing the CoD problem. Indeed, the computation complexity of QMC is linear to the number of samples, allowing more efficient simulation of quantum systems than conventional methods \citep{kummel2003biography}.

Despite the success on several quantum domains, it remains challenging to apply QMC to fermion systems, due to the well-known \textit{sign problem} of fermion systems \citep{pan2022sign}.
Specifically, physics imposes that the wave function of a fermion system is anti-symmetric (AS), meaning that the function value will be flipped when the positions of two electrons inside the system are swapped. This property, when combined with the continuity of the wave function, results in a complicated zero set (set of input with zero function value), in which case the computation cost of conventional QMC can grow exponentially \citep{PhysRevB.92.045110}. 
% Probability-based approaches such as QMC struggle to solve systems with complicated zero sets, and their computation cost can grow exponentially \citep{PhysRevB.92.045110}.
% While successful for systems with nodeless ground state functions, applying QMC to fermion systems introduces another significant challenge, the sign problem. 
% The sign problem means that the complex sign structure of the ground state wave function can exponentially increase the computational cost of most QMC methods\citep{PhysRevB.92.045110}. 
To address the problem, one has to confine Monte Carlo walkers from crossing the zero set \citep{anderson1975random}. The performance of this method largely relies on the prediction of the zero set, which is known to be difficult. 
% The fix-node approximation was introduced to mitigate this problem, which confines Monte Carlo walkers to the region defined by an anti-symmetric trial wave function. 
% As a result, the accuracy of the QMC method heavily relies on the quality of the trial wave function.

Recently, neural networks (NNs) have demonstrated remarkable expressivity in approximating functions in image understanding and language modeling \citep{krizhevsky2012imagenet,vaswani2017attention}. The potential of revolutionizing quantum physics by representing wave functions as NNs has also been demonstrated \citep{carleo2017solving, choo2020fermionic, hibat2020recurrent, pfau2020ab, hermann2020deep, pescia2022neural, gnech2022nuclei, von2022self, li2023forward}. In particular, \citet{pfau2020ab} developed FermiNet that models wave function via equivariant NNs and the Slater determinants, showing promising results in molecules with small sizes of electron $N$. 
%a wave function ansatz. It utilizes the capacity of equivariant NNs to capture the correlation of particles. To ensure that the network is AS, they use the Slater determinant to model the wave function, showing promising results in small molecules.

% \lrc{Concern: it is hard to say that those are primary challenges for the FermiNet, especially for the latter, which we have proven that it can be continuous. So I added `from the theoretical perspective'}

However, the $\mathcal O(N^3)$ computational complexity of Ferminet hinders its application to large-scale systems. 
% The computation of wave functions through Slater determinants requires a time complexity of $\mathcal{O}(N^3)$, where $N$ represents the number of electrons in fermion systems. 
This is exacerbated by the requirement to compute the second-order derivative in the \sch equation and the gradient of neural network parameters.  Indeed, \cite{huang2023geometry,zweig2022towards} shows that an exponential number ($\Omega(\exp(N))$) of Slater determinants
is needed to achieve universal approximation.
In pursuit of more effective and efficient architectures, \citet{hutter2020representing} delves into the fundamental problems related to the representation of antisymmetry continuous (ASC) functions. They elucidate that while ASC functions in one-dimensional systems can be represented by a single determinant, the adequacy and requisite number of determinants for ASC functions in the three-dimensional space remain an open problem. %\yht{already deleted two sentence here. Check}
% In response to the problem, \citet{han2019solving} introduces an $\mathcal{O}(N^2)$ ansatz to represent AS functions, but discontinuous approximators are used in their representation. 
% \citet{pang2022n} proposes an $\mathcal{O}(N^2)$ ansatz capable of representing the ASC wave function for only the ground state of fermion systems. Their method relies on the tiling property \cite{ceperley1991fermion} of the ground state, and cannot generalize to universal ASC functions.

% two primary challenges \lrc{from the theoretical perspective}. Firstly, its reliance on the Slater determinant leads to $\mathcal{O}(N^3)$ computational complexity for an $N$-electron system, limiting scalability. Secondly, while \citet{hutter2020representing} demonstrates that continuous equivariant functions can represent all continuous anti-symmetric functions for one-dimensional systems or when $N \leq 2$, the feasibility for higher dimensions ($d > 1$) and more particles ($N > 2$) remains unresolved. 

In this work, we propose a novel ansatz that represents \textit{any} ASC function in any dimensional space, and its computation complexity is $\tilde{\mathcal{O}}(N^2)$, where $\tilde{\mathcal O}$ hides logarithmic terms.
Our method is based on permutation-equivariant functions, and can be easily represented by modern neural networks.
We further demonstrate that our method can be converted into the form of determinants, thus addressing the open problem proposed by \citet{hutter2020representing} that $\mathcal{O}(N)$ generalized Slater determinants are sufficient to represent any ASC functions in high dimensional spaces.
From a theoretical side, the key lemma we prove can further help improve the efficiency of \citet{chen2023exact} from \(\mathcal{O}(N^5)\) to $\mathcal O(N^4)$, and our ansatz generalizes beyond \citet{pang2022n} that represents only ground state wave functions. 
We also shed light on how to leverage our architecture in practice to achieve significant acceleration.
% Our paper presents an $\tilde{\mathcal{O}}(N^2)$ ansatz using permutation-equivariant multi-electron functions, establishing it as a universal continuous anti-symmetric approximator. This greatly surpasses the previously most efficient $\mathcal{O}(N^5)$ ansatz by \citep{chen2023exact}. Additionally, we resolve the open problem from \citep{hutter2020representing}, showing that $\mathcal{O}(N)$ Slater determinants can achieve universal continuous anti-symmetric representation.

\section{Preliminary}

% \yht{Please use $\mathbf X$ for matrix.}
% In this section, we formally introduce the concepts of ASC functions, as well as necessary definitions for our architecture and proof.
Define \(X_i \in \mathbb{R}^{d}\) as the spatial coordinates of a electron in the standard \(d\)-dimensional Euler space, and \(\mathbf X = (X_1^\top, \ldots, X_N^\top) \in \mathbb{R}^{N \times d}\) as the configuration of an \(N\)-electron system. In quantum mechanics, the wave function \(\Psi: \mathbb R^{N\times d} \mapsto \mathbb C\) encodes all information of a system, and $|\Psi(X_1, \ldots, X_N)|^2$ represents the probability that $N$ electrons are in the position of $(X_1, \ldots, X_N)$.%\footnote{For simplicity, we omit the spin property of electrons in our paper. Notice that all results can be easily applied to systems with spin included.}

\paragraph{Anti-Symmetry.} Following Fermi-Dirac statistics, the wave function of any fermion system should be anti-symmetric. Specifically, we have \(\Psi(X_1, \ldots, X_N) = \sigma(\pi) \Psi(X_{\pi(1)}, \ldots, X_{\pi(N)})\) for any permutation \(\pi\), where \(\sigma(\pi) = \pm 1\) denotes the sign of the permutation $\pi$. In addition, the wave function should also be continuous. For abbreviation, we use ASC to represent ``anti-symmetric and continuous''. Notice that for any AS function $\Psi$, we have $\Psi(\mathbf X) = 0$ for any $\mathbf X$ in $\Omega \triangleq \{\mathbf{X} \in \mathbb{R}^{N \times d} : \exists i \neq j, X_i = X_j\}.$ We call $\Omega$ the intrinsic zero set, which is a subset of the zero set $\{\mathbf X: \Psi(\mathbf X) = 0\}$.

\paragraph{(Generalized) Slater determinant.} Conventional way to represent the AS function is to use the Slater determinant \citep{foulkes2001quantum}. Specifically, For $N$ chosen single-particle wave functions $\{\psi_1, \psi_2, \ldots, \psi_N\}$, $\Psi(\mathbf X) = \det (\mathbf \Psi)$, where $\mathbf \Psi$ is a square matrix, and $\mathbf \Psi_{ij} = \psi_i(X_j)$.
% \begin{equation}
% \Psi(X) = 
% \begin{vmatrix}
% \psi_1(X_1) & \psi_2(X_1) & \cdots & \psi_N(X_1) \\
% \psi_1(X_2) & \psi_2(X_2) & \cdots & \psi_N(X_2) \\
% \vdots      & \vdots      & \ddots & \vdots      \\
% \psi_1(X_N) & \psi_2(X_N) & \cdots & \psi_N(X_N) 
% \end{vmatrix}.
% \end{equation}
% each row of square matrix $\mathbf \Psi$ to a particle $X_i$, and each column corresponds to a wave function $\psi_i(\cdot)$.  
The anti-symmetry is ensured by the property of determinants, as $\det(\mathbf \Psi)$ goes to its negative when two rows are swapped. To leverage the expressivity of neural networks, \citet{pfau2020ab} proposes a generalization of the Slater determinant. In their approach, the single-particle wave function $\psi_i(X_j)$ is replaced with $\psi_i(X_j|X_{\neq j})$, i.e.,
\begin{equation}\label{eq:generalized_slater}
\Psi(X) = 
\begin{vmatrix}
\psi_1(X_1|X_{\neq 1}) & \cdots & \psi_N(X_1|X_{\neq 1}) \\
\psi_1(X_2|X_{\neq 2}) & \cdots & \psi_N(X_2|X_{\neq 2}) \\
\vdots         & \ddots & \vdots      \\
\psi_1(X_N|X_{\neq N}) & \cdots & \psi_N(X_N|X_{\neq N}) 
\end{vmatrix}.
\end{equation}
Here, $X_{\neq j}$ denotes the configuration of all electrons except electron $j$. Importantly, $\psi_j(X_i|X_{\neq i})$ is designed to be symmetric to the coordinates in $X_{\neq j}$. This modification preserves the AS property and provides a more expressive representation of electron interactions.

\paragraph{Representing ASC Function}
To explore more effective and efficient representation, \citet{hutter2020representing} studies the expressivity of the generalized Slater determinants, and shows that for $d=1$, \cref{eq:generalized_slater} can represent any ASC functions when $\psi_i(X_j|X_{\neq j})$ is continuous and symmetric to $X_{\neq j}$.
% This is proven by defining a total order ``$<$'' on $\mathbb{R}^d$ when $d=1$ and use this order to index particles. 
Specifically, let $\pi \in S_N$ be the permutation that sorts scalar $X_i$ in the real axis. The single-particle wave function $\psi_i(X_j | X_{\neq j})$ is defined as:
\begin{equation*}
\psi_i(X_j | X_{\neq j}) := 
\begin{cases} 
|\Psi(X_{\pi(1)}, \ldots, X_{\pi(N)})|^{1/N} & \text{if } j = \pi(i), \\
0 & \text{else}.
\end{cases}
\end{equation*}
% This construction implies that $\Psi(X)$ becomes the determinant of a permuted diagonal matrix, with the $i$-th column having a non-zero entry only at row $\pi(i)$. 
This construction is valid when $\Psi(X_{\pi(1)}, \ldots, X_{\pi(N)}) \geq 0$. For the general case, one can modify this approach by replacing $|\Psi|^{1/N}$ in $\psi_1$ with $\text{sign}(\Psi)|\Psi|^{1/N}$.
% thereby accommodating any sign of $\Psi(X_{\pi(1)}, \ldots, X_{\pi(N)})$.

The continuity of the construction is derived in \citet{hutter2020representing}. Intuitively, the permutation \(\pi\) changes when the order of two inputs is swapped. In one-dimensional space, this means \(X_i = X_j\) for some \(i \neq j\) and thus $\Psi(\mathbf X) = 0$. Therefore, there is no discontinuity between the first case and the second case in the construction.
Conversely, when \(d > 1\), \(\Psi\) may be non-zero when the order of inputs changes, in which case $\psi_i$ is discontinuous and can be hard to represent by neural networks. 
% Consequently, their ability to express discontinuous aspects of the function is limited.

Recently, \citet{pang2022n} designed an ansatz that represents the ASC wave function for all \textit{grouth state} systems. Their construction leverages the tiling property of the ground state and cannot be generalized to general ASC functions. In this paper, we answer open question from \citet{hutter2020representing}:
% To address this limitation, FermiNet employs a strategy of linearly combining multiple determinants to enhance its accuracy. Despite its practical effectiveness, the theoretical foundation of this approach, especially in the context of general \(N\) and \(d\), remains unresolved. This leads to the following open question by \citep{hutter2020representing}:
\\\textit{How to represent general ASC function via the generalized Slater determinants for \(N\geq 2, d\geq 2\)?
}
\section{\(\Tilde{\mathcal{O}}(N^2)\) Continuous Anti-symmetric Representation \label{ansatz}}

% In the realm of quantum mechanics, particularly for systems represented by neural networks, it is crucial to have a continuous representation for continuous functions {why?} 
% In this section, we introduce a novel ansatz to represent all ASC functions with a computational complexity of \(\tilde {\mathcal{O}}(N^2)\). We use our ansatz to answer the above open question in the next section. 

\paragraph{AS basis.}
As discussed above, the difficulty in representing ASC function in high dimension spaces is that for any order defined in $\mathbb R^d$, the indices of $\mathbf X$ can change even when $\Psi(\mathbf X)$ is not zero. We get rid of representation via total order and choose to decompose a ASC function into multiple pieces and represent each piece by one basis. Specifically, for $y \in\mathbb R^d$, we define
\begin{equation*}
\Tilde{f}_y(\mathbf X) = \prod_{i<j} y^T (X_i - X_j).
\end{equation*}
Obviously, for any $y$ we have $\tilde f_{y}$ is ASC. The lemma below is essential to our result.
\begin{lemma}{(Support Set)}
\label{support}
For any \(K \geq dN + 1\), there exist vectors \(y_1, \ldots, y_K \in \mathbb{R}^d\) such that the intersection of the zero set of \(\Tilde{f}_{y_k}(X)\) equals $\Omega$, i.e.
\begin{equation}\label{eq:zero}
    \bigcap_{k=1}^K \{\mathbf X:  \tilde f_{y_k}(\mathbf X) = 0\} = \Omega.
\end{equation}
Furthermore, if $y_k$ is uniformly selected from $\{y \in \mathbb R^d:\|y\|=1\}$, then \cref{eq:zero} holds almost surely, i.e. the probability it holds is $1$.
\end{lemma}

\Cref{support} ensures that when $K \geq dN+1$, for any ASC function $\Psi$ and any $\mathbf X$ such that $\Psi(\mathbf X) \neq 0$, there exists at least one index $k$ such that $f_{y_k}(\mathbf X) \neq 0$. This property is highly desired as we can symmetric functions on $f_k$ to recover the value of $\Psi(\mathbf X)$.
As we will show, $K$ corresponds to the number of determinants required to represent any ASC function.

The proof is closely related to differentiable manifold and probability theory, and we defer it to \cref{appendix}. 
Naively, for $K \geq \mathcal O(N^2d)$, one can easily prove the lemma according to the principle of pigeonholes. \Cref{support} demonstrates that $\mathcal O(Nd)$ functions are sufficient.
%, and it links the construction to dimension $d$. 
We believe this lemma is of independent technical interest.

% For  vectors \(y_1, \ldots, y_K \in \mathbb{R}^d\) satisfying \cref{support}, the function \(f_k(X)\) is defined as:
\paragraph{Continuous anti-symmetric ansatz.}
We propose the following wave function Ansatz:
\begin{equation}\label{eq:ansatz}
\phi(\mathbf X) = \sum_{k=1}^{K} f_k(\mathbf X)g_k(\mathbf X),
\end{equation}
where \(\{f_k(\mathbf X)\}_{k=1}^K\) are pre-defined AS functions, and \(\{g_k(\mathbf X)\}_{k=1}^K\) are general continuous symmetric functions to be simulated by neural networks. Compared to the Jastrow ansatz presented in \cite{zweig2022towards}, \cref{eq:ansatz} allows for a more powerful representation of the AS function, as shown below:%We now present our main theorem below.

% \begin{lemma}{(Anti-symmetric)}
% The function \(\phi(X)\) is anti-symmetric under the permutation of any two elements \(X_i\) and \(X_j\) in \(X\).
% \end{lemma}

% which is well defined as $\sum_{i=1}^{K} \Tilde{f}_i^2(X) = 0 \Leftrightarrow X \in \Omega$.

\begin{theorem}{(Continuous Universality)}
Let \(K = dN + 1\). There exists $K$ pre-defined functions $f_1,\ldots, f_K$, such that for any anti-symmetric continuous function \(\Psi\),  we can find $K$ continuous symmetric functions \(g_1, \ldots, g_K \in \mathcal{C}(\mathbb{R}^{N \times d})\) such that $\phi$ in \cref{eq:ansatz} equals \(\Psi\).
\end{theorem}

% The proof is succinct and important for understanding our construction, and we attach it below.
\begin{proof}
According to \cref{support}, we choose $y_1,\ldots, y_K \in \mathbb R^d$ such that 
$   \bigcap_{k=1}^K \{\mathbf X:  \tilde f_{k}(\mathbf X) = 0\} = \Omega.
$
We construct the normalized basis as
\begin{equation}
\label{deff}
f_k(\mathbf X) = 
\begin{cases}
    \dfrac{\Tilde{f}_{k}(\mathbf X)}{\sqrt{\sum_{i=1}^{K} \Tilde{f}_{i}^2(\mathbf X)}}, & \text{if } \mathbf X \notin \Omega, \\
    0, & \text{if } \mathbf X \in \Omega.
\end{cases}
\end{equation}
Here $\tilde f_k$ stands for $\tilde f_{y_k}$ for simplicity. Notice that $f_k$ is well-defined, since for all $\mathbf X \notin \Omega, \sum_{i=1}^K \tilde f_i^2(\mathbf X) > 0$. In addition, we always have $|f_k(\mathbf x)|\leq 1$.

To represent $\Psi$, we set \(g_k(\mathbf X) = f_k(\mathbf X)\Psi(\mathbf X)\). To prove that \(g_k(\mathbf X)\) is continuous, we only need to consider $\mathbf X \in \Omega$ where $f_k$ is discontinuous. Since \(\Psi(\mathbf X)=0\) for $\mathbf X \in \Omega$, $\Psi$ is continuous, and \(|f_k(\mathbf X)| \leq 1\) is bounded, for $\mathbf X \in \Omega, \lim_{\mathbf X' \rightarrow \mathbf  X} g_k(\mathbf X') = 0 = g_k(\mathbf X)$. % This completes the proof.
\end{proof}

\paragraph{Computation of the construction.} 
In our construction, $f_k$ is irrelevant to the solution $\Psi$, and \cref{support} implies that $y_k$ can be randomly selected and fixed beforehand. Given a vector $y_k$, the complexity of computing $f_k$ is $\mathcal O(N\log^2 N)$ \citep{doi:10.1080/00207160008804911}, so the total computation cost of $\{f_1,\ldots, f_K\}$ is $\mathcal O(N^2d \log^2 N)$. In addition, we only need to use NNs to simulate the continuous $g_k$. Computing $K$ symmetric features $g_k$ in FermiNet \citep{pfau2020ab} or PsiFormer \citep{von2022self} has a complexity of \(\mathcal{O}(N^2)\). Altogether, \cref{eq:ansatz} can be computed with a cost of $\mathcal O(N^2 d\log^2 N)$ in practice.  
% This implies that the cost for calculating \(g_1(X), \ldots, g_K(X)\) also falls within \(\mathcal{O}(N^2)\). Similarly, the computational effort for evaluating \(y_k^\top X_i\) for all \(k \in [K]\) and \(i \in [N]\) is \(O(N^2)\).

% Furthermore, as highlighted by \citep{doi:10.1080/00207160008804911}, the square of the Vandermonde determinant can be efficiently computed using the expression:
% \begin{equation*}
% \left|\prod_{i=1}^N \left[\frac{d}{dx}\prod_{j=1}^N(x_j-x)\right]_{x=x_i}\right|
% \end{equation*}
% which requires \(\mathcal{O}(N\log^2_2N)\) arithmetic operations. Hence, the evaluation of \(K\)\lrc{K or N?} such Vandermonde determinants demands \(\mathcal{O}(N^2\log^2_2N)\) computations. Consequently, the overall computational cost for implementing \(\phi(X)\) is \(\Tilde{\mathcal{O}}(N^2)\).

\paragraph{Representation via Determinant Ansatz}
We now demonstrate our construction \cref{eq:ansatz} can be used for the generalized Slater determinants structure in \cref{eq:generalized_slater}.
% Slater determinants are a prevalent method for representing wavefunctions in quantum mechanics. \citet{hutter2020representing} demonstrated that for \(d=1\) or \(n=2\), any continuous anti-symmetric function can be represented as a Slater determinant. However, the case for \(n>2\) and \(d>1\) remained unresolved. 
Specifically, \(\mathcal O(N)\) determinants are sufficient to represent any ASC function.

\begin{theorem}\label{thm:open}
For any anti-symmetric continuous function \(\Psi(\mathbf X)\), there exist \(K=dN+1\) determinants $\Phi^1,\ldots,\Phi^K$ where the elementary function \(\phi_i^k(x_j|x_{-j})\) is continuous and symmetric to \(X_{\neq j}\), such that
$\Psi = \sum_{k=1}^{K} \Phi^k$.
\end{theorem}
The proof can be found in \cref{app:determinant}, and we briefly discuss the intuition here. According to \cref{support}, $\Psi$ can be decomposed into the sum of $K$ ASC functions $\Psi^k \triangleq f_k^2 \Psi$, and we simulate each $\Psi^k$ using one determinant. The construction of $f_k$ implies that $\Psi^k(\mathbf X) = 0$ when $y_k^\top X_i = y_k^\top X_j$. This allows us to ``project'' a $d$ dimensional vector into one-dimensional space where the order can be defined along this direction as if $d=1$, and the construction in \cite{hutter2020representing} can be reused.

Remarkably, \cref{thm:open} proves that when the elementary function $\phi_i^k$ is continuous and symmetric to $X_{\neq j}$, our construction can represent any ASC function, answering the open problem proposed by \cite{hutter2020representing} that $K = Nd+1$ determinants are sufficient. It remains unclear whether $K$ determinants are \textit{necessary} for the representation. In fact, since the construction in \cite{hutter2020representing} is quite sparse with most of the functions being $0$, we believe that reducing the number of $K$ should be possible.

% \paragraph{Comparison with \citep{chen2023exact}}
% Notice that recently, 
Recently, \cite{chen2023exact} use a composite function ansatz $\phi(\mathbf X) = g(\pmb{\eta}(\mathbf X))$ to represent ASC function, where \(g\) is a continuous odd function, and \(\pmb{\eta}: \mathbb{R}^{dN} \rightarrow \mathbb{R}^m\) is a set of ASC functions. 
Their theorem only applies to $\Psi$ defined over a compact set, and the nature of the composite function makes it hard to accommodate the determinant structure.
% Their ansatz is expressed as a combination of an odd function with a fixed set of basis functions.
Interestingly, \citep{chen2023exact} proves a lemma (Lemma 2.5) that the intersection of the zero set of $\mathcal O(N^2)$ AS functions can be $\Omega$, which is a strictly weaker version of \cref{support}. This also results in their high requirement of AS basis, with $m$ having the order of at least \(\mathcal{O}(N^5)\).
% However, their approach is less efficient in terms of the required number of basis functions. 
% They necessitate \(K\) to be \(\mathcal{O}(N^2)\), leading to a larger set of basis functions. 
By applying our lemma, the cardinality of \(m\) could be directly reduced to \(\mathcal{O}(N^4)\) without any change of the ansatz. This again demonstrates the technical novelty and imporatance of \cref{support}, and also illustrates the efficiency of our construction over other ansatzes.

\section{Conclusion}
In conclusion, we propose a novel ansatz to represent ASC functions. Our ansatz is easily implemented, computationally efficient, and universally capable. 
Our work also addresses the open question of representing ASC functions using the generalized Slater determinants, without being impeded by its slow computation. We sincerely hope that the ansatz can be practically implemented in the future and be beneficial to the development of quantum mechanics.
% This advancement has significant implications for improving the computational efficiency in quantum many-body system simulations.

\bibliography{iclr2024_conference}

\begin{thebibliography}{30}
\providecommand{\natexlab}[1]{#1}
\providecommand{\url}[1]{\texttt{#1}}
\expandafter\ifx\csname urlstyle\endcsname\relax
  \providecommand{\doi}[1]{doi: #1}\else
  \providecommand{\doi}{doi: \begingroup \urlstyle{rm}\Url}\fi

\bibitem[Anderson(1975)]{anderson1975random}
James~B Anderson.
\newblock A random-walk simulation of the schr{\"o}dinger equation: H+ 3.
\newblock \emph{The Journal of Chemical Physics}, 63\penalty0 (4):\penalty0 1499--1503, 1975.

\bibitem[Aznabaev et~al.(2018)Aznabaev, Bekbaev, and Korobov]{aznabaev2018nonrelativistic}
DT~Aznabaev, AK~Bekbaev, and Vladimir~I Korobov.
\newblock Nonrelativistic energy levels of helium atoms.
\newblock \emph{Physical Review A}, 98\penalty0 (1):\penalty0 012510, 2018.

\bibitem[Carleo \& Troyer(2017)Carleo and Troyer]{carleo2017solving}
Giuseppe Carleo and Matthias Troyer.
\newblock Solving the quantum many-body problem with artificial neural networks.
\newblock \emph{Science}, 355\penalty0 (6325):\penalty0 602--606, 2017.

\bibitem[Chen \& Lu(2023)Chen and Lu]{chen2023exact}
Ziang Chen and Jianfeng Lu.
\newblock Exact and efficient representation of totally anti-symmetric functions.
\newblock \emph{arXiv preprint arXiv:2311.05064}, 2023.

\bibitem[Choo et~al.(2020)Choo, Mezzacapo, and Carleo]{choo2020fermionic}
Kenny Choo, Antonio Mezzacapo, and Giuseppe Carleo.
\newblock Fermionic neural-network states for ab-initio electronic structure.
\newblock \emph{Nature communications}, 11\penalty0 (1):\penalty0 2368, 2020.

\bibitem[Feynman et~al.(1965)Feynman, Leighton, and Sands]{feynman1965feynman}
Richard~P Feynman, Robert~B Leighton, and Matthew Sands.
\newblock The feynman lectures on physics; vol. 3.
\newblock \emph{American Journal of Physics}, 33\penalty0 (9):\penalty0 750--752, 1965.

\bibitem[Foulkes et~al.(2001)Foulkes, Mitas, Needs, and Rajagopal]{foulkes2001quantum}
WMC Foulkes, Lubos Mitas, RJ~Needs, and Guna Rajagopal.
\newblock Quantum monte carlo simulations of solids.
\newblock \emph{Reviews of Modern Physics}, 73\penalty0 (1):\penalty0 33, 2001.

\bibitem[Gnech et~al.(2022)Gnech, Adams, Brawand, Carleo, Lovato, and Rocco]{gnech2022nuclei}
Alex Gnech, Corey Adams, Nicholas Brawand, Giuseppe Carleo, Alessandro Lovato, and Noemi Rocco.
\newblock Nuclei with up to a= 6 nucleons with artificial neural network wave functions.
\newblock \emph{Few-Body Systems}, 63\penalty0 (1):\penalty0 7, 2022.

\bibitem[Han et~al.(2018)Han, Jentzen, and E]{han2018solving}
Jiequn Han, Arnulf Jentzen, and Weinan E.
\newblock Solving high-dimensional partial differential equations using deep learning.
\newblock \emph{Proceedings of the National Academy of Sciences}, 115\penalty0 (34):\penalty0 8505--8510, 2018.

\bibitem[Heifetz(2020)]{heifetz2020quantum}
Alexander Heifetz.
\newblock \emph{Quantum mechanics in drug discovery}.
\newblock Springer, 2020.

\bibitem[Hermann et~al.(2020)Hermann, Sch{\"a}tzle, and No{\'e}]{hermann2020deep}
Jan Hermann, Zeno Sch{\"a}tzle, and Frank No{\'e}.
\newblock Deep-neural-network solution of the electronic schr{\"o}dinger equation.
\newblock \emph{Nature Chemistry}, 12\penalty0 (10):\penalty0 891--897, 2020.

\bibitem[Hibat-Allah et~al.(2020)Hibat-Allah, Ganahl, Hayward, Melko, and Carrasquilla]{hibat2020recurrent}
Mohamed Hibat-Allah, Martin Ganahl, Lauren~E Hayward, Roger~G Melko, and Juan Carrasquilla.
\newblock Recurrent neural network wave functions.
\newblock \emph{Physical Review Research}, 2\penalty0 (2):\penalty0 023358, 2020.

\bibitem[Huang et~al.(2023)Huang, Landsberg, and Lu]{huang2023geometry}
Hang Huang, Joseph~M Landsberg, and Jianfeng Lu.
\newblock Geometry of backflow transformation ansatze for quantum many-body fermionic wavefunctions.
\newblock \emph{Communications in Mathematical Sciences}, 21\penalty0 (5):\penalty0 1447--1453, 2023.

\bibitem[Hutter(2020)]{hutter2020representing}
Marcus Hutter.
\newblock On representing (anti) symmetric functions.
\newblock \emph{arXiv preprint arXiv:2007.15298}, 2020.

\bibitem[Iglovikov et~al.(2015)Iglovikov, Khatami, and Scalettar]{PhysRevB.92.045110}
V.~I. Iglovikov, E.~Khatami, and R.~T. Scalettar.
\newblock Geometry dependence of the sign problem in quantum monte carlo simulations.
\newblock \emph{Phys. Rev. B}, 92:\penalty0 045110, Jul 2015.
\newblock \doi{10.1103/PhysRevB.92.045110}.
\newblock URL \url{https://link.aps.org/doi/10.1103/PhysRevB.92.045110}.

\bibitem[Keimer \& Moore(2017)Keimer and Moore]{keimer2017physics}
B~Keimer and JE~Moore.
\newblock The physics of quantum materials.
\newblock \emph{Nature Physics}, 13\penalty0 (11):\penalty0 1045--1055, 2017.

\bibitem[Krizhevsky et~al.(2012)Krizhevsky, Sutskever, and Hinton]{krizhevsky2012imagenet}
Alex Krizhevsky, Ilya Sutskever, and Geoffrey~E Hinton.
\newblock Imagenet classification with deep convolutional neural networks.
\newblock \emph{Advances in neural information processing systems}, 25, 2012.

\bibitem[K{\"u}mmel(2003)]{kummel2003biography}
Hermann~G K{\"u}mmel.
\newblock A biography of the coupled cluster method.
\newblock \emph{International Journal of Modern Physics B}, 17\penalty0 (28):\penalty0 5311--5325, 2003.

\bibitem[Li \& Nakamura(2000)Li and Nakamura]{doi:10.1080/00207160008804911}
Lei Li and Tadao Nakamura.
\newblock Fast parallel algorithms for vandermonde determinants.
\newblock \emph{International Journal of Computer Mathematics}, 73\penalty0 (4):\penalty0 479--486, 2000.
\newblock \doi{10.1080/00207160008804911}.
\newblock URL \url{https://doi.org/10.1080/00207160008804911}.

\bibitem[Li et~al.(2023)Li, Ye, Jiang, Wen, Wang, Li, Li, He, Chen, Ren, et~al.]{li2023forward}
Ruichen Li, Haotian Ye, Du~Jiang, Xuelan Wen, Chuwei Wang, Zhe Li, Xiang Li, Di~He, Ji~Chen, Weiluo Ren, et~al.
\newblock Forward laplacian: A new computational framework for neural network-based variational monte carlo.
\newblock \emph{arXiv preprint arXiv:2307.08214}, 2023.

\bibitem[Needs et~al.(2020)Needs, Towler, Drummond, Lopez~Rios, and Trail]{needs2020variational}
RJ~Needs, MD~Towler, ND~Drummond, Pablo Lopez~Rios, and JR~Trail.
\newblock Variational and diffusion quantum monte carlo calculations with the casino code.
\newblock \emph{The Journal of chemical physics}, 152\penalty0 (15):\penalty0 154106, 2020.

\bibitem[Pan \& Meng(2022)Pan and Meng]{pan2022sign}
Gaopei Pan and Zi~Yang Meng.
\newblock Sign problem in quantum monte carlo simulation.
\newblock \emph{arXiv preprint arXiv:2204.08777}, 2022.

\bibitem[Pang et~al.(2022)Pang, Yan, and Lin]{pang2022n}
Tianyu Pang, Shuicheng Yan, and Min Lin.
\newblock $o(n^{2})$ universal antisymmetry in fermionic neural networks.
\newblock \emph{arXiv preprint arXiv:2205.13205}, 2022.

\bibitem[Pescia et~al.(2022)Pescia, Han, Lovato, Lu, and Carleo]{pescia2022neural}
Gabriel Pescia, Jiequn Han, Alessandro Lovato, Jianfeng Lu, and Giuseppe Carleo.
\newblock Neural-network quantum states for periodic systems in continuous space.
\newblock \emph{Physical Review Research}, 4\penalty0 (2):\penalty0 023138, 2022.

\bibitem[Pfau et~al.(2020)Pfau, Spencer, Matthews, and Foulkes]{pfau2020ab}
David Pfau, James~S Spencer, Alexander~GDG Matthews, and W~Matthew~C Foulkes.
\newblock Ab initio solution of the many-electron schr{\"o}dinger equation with deep neural networks.
\newblock \emph{Physical Review Research}, 2\penalty0 (3):\penalty0 033429, 2020.

\bibitem[Sard(1942)]{sard1942measure}
Arthur Sard.
\newblock The measure of the critical values of differentiable maps.
\newblock 1942.

\bibitem[Toponogov(2006)]{toponogov2006differential}
Victor~A Toponogov.
\newblock \emph{Differential geometry of curves and surfaces}.
\newblock Springer, 2006.

\bibitem[Vaswani et~al.(2017)Vaswani, Shazeer, Parmar, Uszkoreit, Jones, Gomez, Kaiser, and Polosukhin]{vaswani2017attention}
Ashish Vaswani, Noam Shazeer, Niki Parmar, Jakob Uszkoreit, Llion Jones, Aidan~N Gomez, {\L}ukasz Kaiser, and Illia Polosukhin.
\newblock Attention is all you need.
\newblock \emph{Advances in neural information processing systems}, 30, 2017.

\bibitem[von Glehn et~al.(2022)von Glehn, Spencer, and Pfau]{von2022self}
Ingrid von Glehn, James~S Spencer, and David Pfau.
\newblock A self-attention ansatz for ab-initio quantum chemistry.
\newblock \emph{arXiv preprint arXiv:2211.13672}, 2022.

\bibitem[Zweig \& Bruna(2022)Zweig and Bruna]{zweig2022towards}
Aaron Zweig and Joan Bruna.
\newblock Towards antisymmetric neural ansatz separation.
\newblock \emph{arXiv preprint arXiv:2208.03264}, 2022.

\end{thebibliography}
\bibliographystyle{iclr2024_conference}
\newpage
\appendix
\section{Proof of Lemma \ref{support} \label{appendix}}

In preparation for proving Lemma \ref{support}, we establish a foundational understanding by revisiting a fundamental concept in differential topology and introducing a pivotal lemma that addresses the measure of smooth functions when they map into higher-dimensional spaces. This is essential for the subsequent proof.

\textbf{Definition (Critical Point\citep{toponogov2006differential}): } Given a differentiable map \( f: \mathbb{R}^m \to \mathbb{R}^n \), the critical points of \( f \) are the points of \( \mathbb{R}^m \), where the rank of the Jacobian matrix $\mathbf J$ of \( f \) is less than $n$.

The rank of Jacobian matrix gives the dimension of the image of the differential of $f$. At a critical point, the rank of the Jacobian matrix is less than the maximum, which implies that the map $f$ is not injective at these points. 

\begin{lemma}
\label{sard}
Let $m < n$ be two positive integers. Consider a smooth function \(f: \mathbb{R}^m \to \mathbb{R}^{n}\) and a subset \(A \subseteq \mathbb{R}^m\). The image \(f(A)\) is of zero measure in \(\mathbb{R}^{n}\).
\end{lemma}

\begin{proof}
Since $f$ is a mapping into higher-dimensional space, we always have $rank(\mathbf J)\leq m < n$. It implies that every point in $A$ is a critical point, and the image of $f$ consists solely of critical values.
%, and thus must have measure zero in $\mathbb{R}^n$.
Sard's Theorem \citep{sard1942measure} posits that the set of critical values of a smooth function has  Lebesgue measure zero in the target space. This makes the set of critical values ``small" in the sense of a generic property. 
\end{proof}

Now we are ready to prove \cref{support}. For the ease of reference, we restate the \cref{support}: 

\textbf{(Restate) Lemma \ref{support}: \label{support_appendix}}
For any \(K \geq dN + 1\), there exist vectors \(y_1, \ldots, y_K \in \mathbb{R}^d\) such that the intersection of the zero set of \(\Tilde{f}_{y_k}(X)\) equals $\Omega$, i.e. $\bigcap_{k=1}^K \{\mathbf X:  \tilde f_{y_k}(\mathbf X) = 0\} = \Omega.$ Furthermore, if $y_k$ is uniformly selected from $\{y \in \mathbb R^d:\|y\|=1\}$, then \cref{eq:zero} holds almost surely, i.e. the probability it holds is $1$.

For brevity, we define $Y=(y_1, y_2, ..., y_K) \in \mathbb{R}^{dK}$. We say that a configuration \(\mathbf{X}\) \textit{covers} \(Y\) if \(\forall k \in [K], \Tilde{f}_{y_k}(\mathbf{X}) = 0\). Therefore, to prove Lemma \ref{support_appendix}, it suffices to demonstrate that there exists \(Y \in \mathbb{R}^{dK}\) which cannot be covered by any \(\mathbf{X} \notin \Omega\), where \(\Omega = \{\rvX \in \mathbb{R}^{N \times d} : \exists i \neq j, X_i = X_j\}\). Formally, we define 
\[\Theta=\{Y:\exists \rvX\notin\Omega \text{, } Y \textit{is covered by } \rvX\}.\]

In our proof, we will establish that $\Theta$ is a null set in \(\mathbb{R}^{dK}\). Consequently, selecting a \(Y\) uniformly at random from \(\mathbb{R}^{dK}\) will almost surely meet the required condition, thus validating the lemma.

\begin{proof}
Consider selecting \(K\) pairs of numbers from \([N]\) such that for each pair \((i_k, j_k)\), \(i_k < j_k\). The total number of distinct selection strategies is \(T = \left(\frac{N(N-1)}{2}\right)^N\). Let \((i_k^t, j_k^t)\) denote the \(k\)-th pair in the \(t\)-th strategy.

For \(l \in [d]\), define the mapping \(\mathcal{F}_l: \mathbb{R}^d \times \mathbb{R}^{d-1} \rightarrow \mathbb{R}^d\) as
\begin{equation*}
\mathcal{F}_l(x, a)=\left(a_1x_l, \ldots, a_{l-1}x_l, -\sum_{l'=1}^{l-1}a_{l'}x_{l'} - \sum_{l'=l+1}^{d}a_{l'-1}x_{l'}, a_lx_l, \ldots, a_{d-1}x_l\right)^\top.
\end{equation*}

Intuitively, $\mathcal{F}_l(x, a)$ always outputs the vector orthogonal to $x$. Given \(t\in[T]\) and \((l_1, \ldots, l_K) \in [d]^K\) , define the mapping \(\mathcal{A}_t^{(l_1, \ldots, l_K)}: \mathbb{R}^{dN} \times \mathbb{R}^{(d-1)K} \rightarrow \mathbb{R}^{dK}\) such that each column of the output is defined as:
\begin{equation*}
\mathcal{A}_t^{(l_1, \ldots, l_K)}(\rvX, A)_k = \mathcal{F}_{l_k}(X_{i_k^t} - X_{j_k^t}, A_k).
\end{equation*}

We define the union of images as: 
\[\Gamma=\bigcup_{t,l_1,l_2,...l_K}{\text{Image}(\mathcal{A}_t^{(l_1, \ldots, l_K)})}.\]
In the following, we will prove that: a. $\Theta\subseteq \Gamma$. b. $\Gamma$ has zero measure in \(\mathbb{R}^{dK}\).

\textbf{a. $\Theta\subseteq \Gamma$}. We show that for any \(Y\) that can be covered by \(\rvX \notin \Omega\), there exist \(t \in [T]\), \((l_1, \ldots, l_K) \in [d]^K\), \(A \in \mathbb{R}^{(d-1)K}\) such that \(\mathcal{A}_t^{(l_1, \ldots, l_K)}(\rvX, A) = Y\). 

If \(Y\) is covered by \(\rvX\), then there exists $t\in T$, such that \( \forall k, y_k^\top (X_{i_k^t} - X_{j_k^t}) = 0\), i.e., $y_k$ is orthogonal to $x=X_{i_k^t} - X_{j_k^t}$. As \(\rvX\notin\Omega\), we have that for any \(k\), there exists a $l_k\in [d]$ such that \(x_{l_k} \neq 0\). As a result, we can choose \(A_k = \left(\frac{y_{k, 1}}{x_{l_k}}, \ldots, \frac{y_{k, l-1}}{x_{l_k}}, \frac{y_{k, l+1}}{x_{l_k}}, \ldots, \frac{y_{k, d}}{x_{l_k}}\right)^\top\). With detailed calculation, we can prove that: 
\begin{equation*}
\mathcal{F}_{l_k}(x, A_k) = y_k.
\end{equation*}

Selecting appropriate \(t\), \(l_1, \ldots, l_K\)  and \(A\) yields \(\mathcal{A}_t^{(l_1, \ldots, l_K)}(X, A) = Y\).

\textbf{b. $\Gamma$ has zero measure in \(\mathbb{R}^{dK}\)}. For every \(t \in [T]\) and \((l_1, \ldots, l_K) \in [d]^K\), \(\mathcal{A}_t^{(l_1, \ldots, l_K)}\) maps \(\mathbb{R}^{dN+(d-1)K}\) to \(\mathbb{R}^{dK}\). Given that \(K \geq dN + 1\), Lemma \ref{sard} implies that \(\mathcal{A}_t^{(l_1, \ldots, l_K)}\) has zero measure in \(\mathbb{R}^{dK}\). Hence,
\begin{equation*}
m(\Gamma) = m\left(\bigcup_{t=1}^{T} \bigcup_{l_1=1}^{d} \ldots \bigcup_{l_K=1}^{d} \mathcal{A}_t^{(l_1, \ldots, l_K)}(\mathbb{R}^{dN}, \mathbb{R}^{(d-1)K})\right) = 0.
\end{equation*}

Thus, \(\Theta\) also has zero measure, completing the proof.
\end{proof}

\section{Proof for Theorem \ref{thm:open} \label{app:determinant}}
\textbf{(Restate) Theorem \ref{thm:open}}: For any anti-symmetric continuous function \(\Psi(\mathbf X)\), there exist \(K=dN+1\) determinants $\Phi^1,\ldots,\Phi^K$ where the elementary function \(\phi_i^k(x_j|x_{-j})\) is continuous and symmetric to \(X_{\neq j}\), such that
\begin{equation*}
\Psi = \sum_{k=1}^{K} \Phi^k.
\end{equation*}

\begin{proof}
Consider any \(k \in [K]\). Let
\begin{equation*}
\Phi^k(X) = f_k^2(X)\Psi(X),
\end{equation*}
where \(f_k\) is as defined in \cref{deff}. It can be shown that \(\Psi = \sum_{k=1}^{K} \Phi^k\). Here, \(\Phi^k\) is a continuous anti-symmetric function. \(\Phi^k(X)=0\) when \(\exists i \neq j\) such that \(y_k^\top X_i = y_k^\top X_j\).

For any \(k\) and \(X_1, \ldots, X_N\), we can sort \((X_{1}, \ldots, X_{N})\) according to the scalar function \[s_k:\mathbb{R}^{d}\rightarrow \mathbb{R}, s_k(x)=y_k^{\top}x.\]  After sorting, we derive a permutation $\pi$, such that \(y_k^\top X_{\pi(1)} \leq \ldots \leq y_k^\top X_{\pi(N)}\). Then we can define

% \begin{equation*}
% \phi_i^k(X_j | X_{\neq j}) := 
% \begin{cases} 
% |\Phi^k(X_{\pi(1)}, \ldots, X_{\pi(N)})|^{1/n}, & \text{if } j = \pi(i), \\
% 0, & \text{else}.
% \end{cases}
% \end{equation*}
% This forms a permuted diagonal matrix and is valid when $\Phi^k(X_{\pi(1)}, \ldots, X_{\pi(N)}) \geq 0$. Modify this approach by replacing $|\Phi^k|^{1/n}$ in $\phi_1^k$ with $\text{sign}(\Phi^k)|\Phi^k|^{1/N}$ can work in general.

% Hence, when \(\forall l_1 \neq l_2\), \(y_k^\top x_{l_1} \neq y_k^\top x_{l_2}\), each \(\phi_i^k\) is evidently continuous. In the case of \(y_k^{\top}X_{l_1} \rightarrow y_k^{\top}X_{l_2}\), where \(l_1 \neq l_2\),
% \begin{equation*}
% \lim_{y_k^{\top}X_{l_1} \rightarrow y_k^{\top}X_{l_2}} \phi_i^k(X_j|X_{\neq j}) = 0,
% \end{equation*}
% ensuring the continuity of \(\phi_i^k\). This completes the proof.
% \end{proof}

% use psi instread of phi?
\begin{equation*}
\psi_i^k(X_j | X_{\neq j}) := 
\begin{cases} 
\text{sign}(\Phi^k)|\Phi^k(X_{\pi(1)}, \ldots, X_{\pi(N)})|^{1/n}, & \text{if } j = \pi(i), \\
0, & \text{else}.
\end{cases}
\end{equation*}
% This forms a permuted diagonal matrix and is valid when $\Phi^k(X_{\pi(1)}, \ldots, X_{\pi(N)}) \geq 0$. Modify this approach by replacing $|\Phi^k|^{1/n}$ in $\psi_1^k$ with $|\Phi^k|^{1/N}$ can work in general.

Hence, when \(y_k^\top X_{\pi(1)} , \ldots , y_k^\top X_{\pi(N)}\) are distinct, the permutation \(\pi\) is unique, and each \(\psi_i^k\) is evidently continuous. In the case where \(\exists l_1 \neq l_2\) such that \(y_k^{\top}X_{l_1} = y_k^{\top}X_{l_2}\), although the permutation \(\pi\) is not unique, \(\Phi^k(X_{\pi(1)}, \ldots, X_{\pi(N)})\) is \(0\) for all permutations. Thus, we have
\begin{equation*}
\lim_{y_k^{\top}X_{l_1} \rightarrow y_k^{\top}X_{l_2}} \psi_i^k(X_j|X_{\neq j}) =\psi_i^k(X_j|X_{\neq j})|_{y_k^{\top}X_{l_1} = y_k^{\top}X_{l_2}}= 0,
\end{equation*}
ensuring the continuity of \(\psi_i^k\). This completes the proof.
\end{proof}

\end{document}